\documentclass[a4paper,USenglish,cleveref,autoref,thm-restate]{lipics-v2019}

\usepackage{boxedminipage}

\title{A Polynomial Kernel for Line Graph Deletion} 

\author{Eduard Eiben}{Department of Computer Science, Royal Holloway University of London, Egham, United Kingdom}{eduard.eiben@rhul.ac.uk}{https://orcid.org/0000-0003-2628-3435}{}
\author{William Lochet}{Department of Informatics, University of Bergen, Bergen, Norway}{william.lochet@uib.no}{}{Supported by The Bergen Research Foundation (BFS).}

\authorrunning{E. Eiben and W. Lochet}

\Copyright{Eduard Eiben and William Lochet}



\category{}

\relatedversion{Accepted for a publication in the Proceedings of the 28th Annual European Symposium on Algorithms (ESA 2020).}

\supplement{}

\nolinenumbers 

\hideLIPIcs  

\newcommand{\bigO}[1]{\ensuremath{{\mathcal O}\left(#1\right)}}

\newcommand{\hedit}{\textsc{$H$-free-Edge Editing}}

\newcommand{\ldelete}{\textsc{Line-Graph-Edge Deletion}}

\newcommand{\cliqueWitness}{clique partition witness}

\theoremstyle{plain}
\newcommand{\NP}{\textsf{NP}}

\newcommand{\coNP}{\textsf{coNP}}
\newcommand{\NPpoly}{\textsf{NP/poly}}
\newtheorem{rrule}{Reduction Rule}
\newtheorem{observation}[theorem]{Observation}
\newcommand{\CCC}{\mathcal{C}}
\newcommand{\DDD}{\mathcal{D}}
\newcommand{\LLL}{\mathcal{L}}
\newcommand{\NNN}{\mathcal{N}}

\newcommand{\distance}{$\LLL$-distance}
\newcommand{\distC}{\ensuremath{\operatorname{dist}^{\LLL}}}

\begin{document}
	
\maketitle
\begin{abstract}\label{Abstract}
	The line graph of a graph $G$ is the graph $L(G)$ whose vertex set is the edge set of $G$ and there is an edge between $e,f\in E(G)$ if $e$ and $f$ share an endpoint in $G$. A graph is called line graph if it is a line graph of some graph. We study the \ldelete{} problem, which asks whether we can delete at most $k$ edges from the input graph $G$ such that the resulting graph is a line graph. More precisely, we give a polynomial kernel for \ldelete{} with $\bigO{k^{5}}$ vertices. This answers an open question posed by Falk H\"{u}ffner at Workshop on Kernels (WorKer) in 2013.
\end{abstract}

\section{Introduction}\label{Section: Introduction}
For a family $\mathcal{G}$ of graphs, the general $\mathcal{G}$-\textsc{Graph Modification} 
problem asks whether we can modify a graph $G$ into a graph in $\mathcal{G}$ by
performing at most $k$ simple operations. Typical examples of simple operations
well-studied in the literature include vertex deletion, edge deletion, edge addition, or
a combination of edge deletion and addition. We call these problems
$\mathcal{G}$-\textsc{Vertex Deletion}, $\mathcal{G}$-\textsc{Edge Deletion},
$\mathcal{G}$-\textsc{Edge Addition}, and $\mathcal{G}$-\textsc{Edge Editing},
respectively. By a classical result by Lewis and Yannakakis~\cite{LewisY80},
$\mathcal{G}$-\textsc{Vertex Deletion} is \NP-complete for all non-trivial hereditary
graph classes. The situation is quite different for the edge modification problems.
Earlier efforts for edge deletion problems~\cite{ElmallahC88,Yannakakis81}, though
having produced fruitful concrete results, shed little light on a systematic answer, and
it was noted that such a generalization is difficult to obtain.

$\mathcal{G}$-\textsc{Graph Modification} problems have been extensively investigated for
graph classes $\mathcal{G}$ that can be characterized by a finite set of forbidden
induced subgraphs. We say that a graph is $\mathcal{H}$-free if it contains 
none of the graphs in $\mathcal{H}$ as an induced subgraph.
For this special case, the
\textsc{$\mathcal{H}$-free Vertex Deletion} is well understood. If $\mathcal{H}$
contains a graph on at least two vertices, then all of these problems are \NP-complete,
but admit a $c^kn^{\bigO{1}}$ algorithm~\cite{Cai96}, where $c$ is the size of the largest
graph in $\mathcal{H}$ (the algorithms with running time $f(k)n^{\bigO{1}}$ are called
fixed-parameter tractable (FPT) algorithms~\cite{CyganFKLMPPS15,DowneyFellows13}). On
the other hand, the \NP-hardness proof of Lewis and Yannakakis~\cite{LewisY80} excludes
algorithms with running time $2^{o(k)}n^{\bigO{1}}$ under the Exponential Time Hypothesis
(ETH)~\cite{ImpagliazzoP01}. Finally, as observed by Flum and Grohe~\cite{FlumG06} a
simple application of sunflower lemma~\cite{ErdosR60} gives a \emph{kernel} with
$\bigO{k^{c}}$ vertices, where $c$ is again the size of the largest graph in
$\mathcal{H}$. A kernel is a polynomial time preprocessing algorithm which outputs an
equivalent instance of the same problem such that the size of the reduced instance is
bounded by some function $f(k)$ that depends only on $k$. We call the function $f(k)$
the size of the kernel. It is well-known that any problem that admits an FPT algorithm
admits a kernel. Therefore, for problems with FPT algorithms one is interested in
polynomial kernels, i.e., kernels whose size is a polynomial function.

For the edge modification problems, the situation is more complicated. While all of
these problems also admit $c^kn^{\bigO{1}}$ time algorithm, where $c$ is the maximum
number of edges in a graph in $\mathcal{H}$~\cite{Cai96}, the \textsf{P} vs \NP\
dichotomy is still not known. Only recently Aravind et al.~\cite{AravindSS17b} gave the
dichotomy for the special case when $\mathcal{H}$ contains precisely one graph
$H$. From the kernelization point of view, the situation is also
more difficult. The reason is that deleting or adding an edge to a graph can introduce a
new copy of $H$ and this might further propagate. Hence, we cannot use the sunflower
lemma to reduce the size of the instance. Cai asked the question whether
\textsc{$H$-free Edge Deletion} admits a polynomial kernel for all graphs $H$~\cite{bodlaender2006open}. 
Kratsch and Wahlstr{\"{o}}m~\cite{KratschW13} showed
that this is probably not the case and gave a graph $H$ on $7$ vertices such that 
\textsc{$H$-free Edge Deletion} and \textsc{$H$-free Edge Editing} does not admit a
polynomial kernel unless $\coNP\subseteq \NPpoly$. Consequently, it was shown that this 
is not an exception, but rather a rule~\cite{CaiCai15,GuillemotHPP13}. Indeed the result by Cai and 
Cai~\cite{CaiCai15} shows that \textsc{$H$-free Edge Deletion}, \textsc{$H$-free Edge Addition}, 
and \hedit\ do not admit a polynomial kernel whenever $H$ or its complement is a path or a cycle 
with at least $4$ edges or a $3$-connected graph with at least $2$ edges missing. Very recently, Marx and Sandeep~\cite{MarxSandeep20} gave a list of nine graphs, all on $5$ vertices such that if \hedit\ does not admit a kernel for any of these nine graphs under standard complexity assumptions, then \hedit\ admits a polynomial kernel for $|H|\ge 5$ if and only if $H$ is either empty or complete graph. They also provided a similar characterization for \textsc{$H$-free Edge Deletion} and \textsc{$H$-free Edge Editing}. 
This suggests that 
actually the $H$-free edge modification problems with a polynomial kernels are rather rare and only for
small graphs $H$. Recently, Eiben, Lochet, and Saurabh~\cite{EibenLochetSaurabhArXiv} announced a polynomial kernel for the case when $H$ is a paw, which leaves only one last graph on $4$ vertices for which the kernelization of $H$-free edge
modification problems remains open, namely $K_{1,3}$ known also as the claw.

The class of claw-free graphs is a very well studied class of graphs with some interesting algorithmic properties. The most prominent example is probably the algorithm of Sbihi~\cite{SBIHI198053} for computing the maximum independent set in polynomial time. It also has been extensively studied from a structural point of view, and Chudnosky and Seymour proposed, after a series of papers, a complete characterization of claw-free graphs \cite{CHUDNOVSKY2008839}. Because of such a characterization, it seems reasonable to believe that a polynomial kernel for \textsc{Claw-free Edge Deletion} exists. However, the characterization of Chudnosky and Seymour is quite complex, which makes it hard to use. For this reason, as noted by Cygan et al.~\cite{CyganPPLW17}, trying to show the existence of a polynomial kernel in the cases of sub-classes of claw-free graphs seems like a good first step to try to understand this problem. In this paper, we prove the result for the most famous such class, line graphs.

\begin{restatable}{theorem}{polykernelthm}
	\label{thm:main}
\textsc{Line-Graph Edge Deletion} admits a kernel with $\bigO{k^{5}}$ vertices. 
\end{restatable}

\subsubsection*{Overview of the Algorithm} 

As the first step of the kernelization algorithm, we use the characterization of line graphs by forbidden induced subgraphs
to find a set $S$ of at most $6k$ vertices such that for every vertex $v\in S$, $G-(S\setminus \{v\})$ is a line graph. This is simply done by a greedy edge-disjoint packing of forbidden induced subgraphs. Having the set $S$, we use the algorithm by Degiorgi and Simon~\cite{DegiorgiSimon95} to find a partition of edges of $G-S$ into cliques such that each vertex is in precisely $2$ cliques. Let $\mathcal{C} = \{C_1,\ldots, C_q\}$ be the cliques in the partition. Since $G-(S\setminus \{v\})$ is also a line graph, it is a rather simple consequence of Whitney's isomorphism theorem that the neighborhood of $v$ can be covered by constantly many cliques of $\mathcal{C}$. Furthermore, we will show that if a clique $C$ in $\mathcal{C}$ has more than $k+7$ vertices then the optimal solution does not contain an edge in $C$. Hence, we can partition the cliques in $\CCC$ into two groups ``large'' and ``small''. Note that if the optimal solution contains an edge in some small clique $C$, then for this change to be necessary, it has to be propagated from $S$ by modifying small cliques on some clique-path from $S$ to $C$ using only small cliques. We will therefore define the distance of a clique to $S$, without going into too many details in here, to be basically the length of a shortest clique-path from the clique to $S$ using only small cliques. Since there are only $\bigO{|S|}$ cliques in immediate neighborhood of $S$ and the number of cliques in the neighborhood of a small clique is bounded by its size, we obtain that there are at most $\bigO{k^d}$ cliques at distance at most $d$. Our main contribution and most technical part of our proof is to show that we can remove the edges covered by cliques at distance at least $5$ from $G$. This is covered in Section~\ref{sec:boundingDistance}. Afterwards we end up with an instance with all cliques in $\CCC$ at distance at least $5$ from $S$ being singletons. As discussed above there are only $\bigO{k^4}$ cliques at distance at most $4$ and because large cliques stay intact in any optimal solution, it suffices to keep $k+7$ vertices in each large clique, which leads to the desired kernel of size~$\bigO{k^5}$. 

\section{Preliminaries}\label{Section: Preliminaries}
We assume familiarity with the basic notations and terminologies in graph theory. We refer the reader
to the standard book by Diestel~\cite{diestel} for more information. Given a graph $G$ and a set of 
edges $F\subseteq E(G)$, we denote by $G - F$ the graph whose set of vertices is $V(G)$ 
and set of edges is the set $E(G)\setminus F$. 
Given two vertices $u,v\in V(G)$, we let the \emph{distance} between $u$ and $v$ in $G$, denoted $\operatorname{dist}_G(u,v))$, be the number of edges on a shortest path from $u$ to $v$. Furthermore, for $S\subseteq V(G)$ and $u\in V(G)$ we let $\operatorname{dist}_G(u,S)=\min_{v\in S}\operatorname{dist}_G(u,v))$. We omit the subscript $G$, if the graph is clear from the context.

\subparagraph*{Parameterized Algorithms and Kernelization.} For a detailed illustration of the following facts the reader is
referred to~\cite{CyganFKLMPPS15,DowneyFellows13}.
A \emph{parameterized problem} is a language $\Pi \subseteq
\Sigma^*\times \mathbb{N}$, where $\Sigma$ is a finite alphabet; the second
component $k$ of instances $(I,k) \in \Sigma^*\times\mathbb{N}$ is called the
\emph{parameter}. A parameterized problem $\Pi$ is
\emph{fixed-parameter tractable} if it admits a
\emph{fixed-parameter algorithm}, which decides instances $(I,k)$ of
$\Pi$ in time $f(k)\cdot |I|^{\bigO{1}}$ for some computable function
$f$.  

A \emph{kernelization} for a parameterized
problem $\Pi$ is a polynomial-time algorithm that given any instance
$(I,k)$ returns an instance $(I',k')$ such that $(I,k) \in \Pi$ if and
only if $(I',k') \in \Pi$ and such that $|I'|+k'\leq f(k)$ for some
computable function $f$. The function $f$ is called the \emph{size}
of the kernelization, and we have a polynomial kernelization if $f(k)$
is polynomially bounded in $k$. It is known that a parameterized
problem is fixed-parameter tractable if and only if it is decidable
and has a kernelization. However, the kernels implied by this fact are
usually of superpolynomial size.

A \emph{reduction rule} is an algorithm that takes as input an
instance $(I,k)$ of {a parameterized problem $\Pi$} and outputs an instance
$(I',k')$ of the same problem. We say that the reduction rule is
\emph{safe} if $(I,k)$ is a \emph{yes}-instance if and only if $(I',k')$ is a \emph{yes}-instance. In order to describe our kernelization
algorithm, we present a series of reduction rules. 

\subparagraph*{Line graphs.} Given a graph $G$, its \emph{line graph} $L(G)$ is a graph such that each vertex of $L(G)$
represents an edge of $G$ and two vertices of $L(G)$ are adjacent if and only if their
corresponding edges share a common endpoint (are incident) in $G$. It is well known that if the line graphs of two connected graphs $G_1$ and $G_2$ are isomorphic then either $G_1$ and $G_2$ are $K_3$ and $K_{1,3}$, respectively, or $G_1$ and $G_2$ are isomorphic as well~(Whitney's isomorphism theorem~\cite{Whitney32}, see also Theorem~8.3 in~\cite{harary1969graph}). We say that a graph $H$ is a line graph, if there exists a graph $G$ such that $H=L(G)$. Note that in this paper we only consider simple graphs, \emph{i.e.}, the graphs without loops or multiple edges and in particular we also only consider line graphs of simple graphs.  Formally, we then study the following parameterized problem: 

\noindent
\begin{center}
	\begin{boxedminipage}{0.98 \columnwidth}
		\ldelete\\[5pt]
		\begin{tabular}{l p{0.80 \columnwidth}}
			Input: & A graph $G=(V,E)$ and $k\in \mathbb{N}$.\\
			Parameter: & $k$.\\
			Question: & Is there a set of edges $F\subseteq E(G)$ such that $G-F$ is a line graph and $|F|\le k$. 
		\end{tabular}
	\end{boxedminipage}
\end{center}
We call a set of edges $F\subseteq V(G)$ such that $G-F$ is a line graph a \emph{solution} for $G$. A solution $F$ is \emph{optimal}, if there does not exists a solution $F'$ such that $|F'|< |F|$. To obtain our kernel, we will make use of several equivalent characterizations of line graphs. 

\begin{theorem}[see, {\em e.g.,} Theorem~8.4 in \cite{harary1969graph}]\label{thm:LGcharacterization}
	The following statements are equivalent:
	\begin{enumerate}[(1)]
		\item $G$ is a line graph.
		\item The edges of $G$ can be partitioned into complete subgraphs in such a way that no vertex lies in more than two of the subgraphs. 
		\item $G$ does not have $K_{1,3}$ as an induced subgraph, and if two odd triangles (triangles with the property that there exists another vertex adjacent to an odd number of triangle vertices) share a common edge, then the subgraph induced by their vertices is $K_4$. 
		\item None of nine graphs of Figure~\ref{fig:forbiddenCharacterization} is an induced subgraph of $G$. 
\end{enumerate}
\end{theorem} 
\begin{figure}
		\centering
	\includegraphics[width=.9\textwidth]{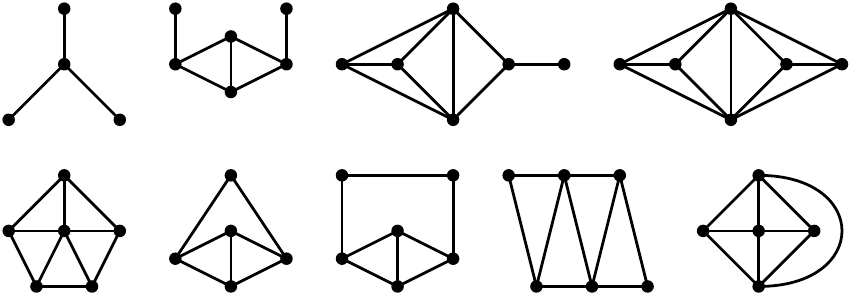}
	\caption{The nine minimal non-line graphs, from characterization of line graphs by forbidden induced subgraphs of Beineke~\cite{Beineke70}. Note that all of these graphs have at most $6$ vertices.}
	\label{fig:forbiddenCharacterization}
\end{figure}

\section{Structure of Line Graphs}

To obtain our kernel, we heavily rely on different characterizations of line graphs given by Theorem~\ref{thm:LGcharacterization}. The two main characterizations used throughout the paper are given in points (2) and (4) To ease the presentation of our techniques, we will define a notion of a \emph{\cliqueWitness} for $G$, whose existence is implied by the point (2) of Theorem~\ref{thm:LGcharacterization}.
Let $G$ be a line graph, a \emph{\cliqueWitness} for $G$ is a set $\mathcal{C} = \{C_1,\ldots, C_q\}$ be such that: 
\begin{itemize}
	\item $C_i\subseteq V(G)$ for all $i\in [q]$,
	\item $G[C_i]$ is a complete graph for all $i\in [q]$, that is every $C_i$ is a clique in $G$,
	\item $|C_i\cap C_j|\le 1$ for all $i\neq j\in [q]$,
	\item every $v\in V(G)$ is in exactly two sets in $\mathcal{C}$, and 
	\item for every edge $uv\in E(G)$ there exists exactly one set $C_i\in \mathcal{C}$ such that $\{u,v\}\subseteq C_i$.
\end{itemize}

Note that by Theorem~\ref{thm:LGcharacterization}, $G$ is a line graph if and only if there exists a \cliqueWitness\ for $G$. The following three observations follow directly from the definition of \cliqueWitness\ and will be useful throughout the paper.

\begin{observation}\label{obs:MaximalCliques}
	If $\CCC$ is \cliqueWitness\ for $G$ then every clique in $\CCC$ is either a singleton, $K_2$, or a maximal clique in $G$. 
\end{observation}

\begin{observation}\label{obs:MaximalCliques2}
	If $\CCC$ is \cliqueWitness\ for $G$, then every maximal clique in $G$ of size at least $4$ is in $\CCC$. 
\end{observation}

\begin{observation}\label{obs:MaximalCliques3}
	If $\CCC$ is \cliqueWitness\ for $G$, then any clique of $G$ which is not a sub-clique of some element of $\CCC$ is a triangle. 
\end{observation}

We would like to point out that given a line graph $G$ one can find a \cliqueWitness\ for $G$ for example by using an algorithm of Degiorgi and Simon~\cite{DegiorgiSimon95} for recognition of line graphs in polynomial time. In the following lemma, we sketch the main procedure of their algorithm together with necessary modifications to actually output a \cliqueWitness\ instead of the underlying graph $H$ such that $G=L(H)$, for completeness.

\begin{lemma}
	\label{lem:constructingWitness}
	Given a graph $G$, there is an algorithm that in time $\bigO{|E(G)|+|V(G)|}$ decides whether $G$ is a line graph and if so, constructs a \cliqueWitness\ for $G$.
\end{lemma}

\begin{proof}
	The algorithm by Degiorgi and Simon construct the input graph $G$ by adding vertices one at a time, at each step it chooses a vertex to add that is already adjacent to at least one previously-added vertex. That is it construct graphs $G_1$, $G_2$, $\ldots$, $G_n=G$ such that $G_i$ is a connected subgraph of $G$ on $i$ vertices. At each step it maintains a graph $H_i$ such that $G_i$ is a line graph of $H_i$. In here, we can actually keep a \cliqueWitness\ $\CCC_i$ for $G_i$ such that there is a bijection $\varphi_i$ between vertices of $H_i$ and clique in $\CCC_i$ such that $uv\in E(H_i)$ if and only if $|\varphi_i(u)\cap \varphi_i(v)|= 1$. 
	
	The algorithm heavily relies on the Whitney's isomorphism theorem that implies that if the underlying graph of $G_i$ has at least $4$ vertices, then the underlying graph $H_i$ is unique up to isomorphism.
	When adding a vertex $v$ to a graph $G_i$ for $i\le 4$, the algorithm simply brute-forces the possibilities for $H_i$ and $\CCC_i$.
	
	When adding a vertex $v$ to $G_i$ when $i>4$, let $S$ be the subgraph of $H_i$ formed by the edges that correspond to the neighbors of $v$ in $G_i$. 
	Check that $S$ has a vertex cover consisting of one vertex or two non-adjacent vertices, \emph{i.e.}, there are cliques $C_1$ and $C_2$ in $\CCC_i$ with $C_i\cap C_2 = \emptyset$ and $S\subseteq C_1\cap C_2$. If there are two vertices in the cover, add an edge (corresponding to $v$) that connects these two vertices in $H_i$ and add $v$ to both $C_1$ and $C_2$. If there is only one vertex $u$ in the cover, then add a new vertex to $H_i$, adjacent to this vertex, add $v$ to the clique $\varphi_i(u)$ in $\CCC_i$ and add a new clique $\{v\}$ to $\CCC_{i}$ to create $\CCC_{i+1}$.
\end{proof}

\subsection{Level Structure of Instances}

For the rest of the paper, let $G$ be the input graph and let $S$ be a set of at most $6k$ vertices such that for every $v\in S$ the graph $G-(S\setminus \{v\})$ is a line graph. 
We let $\mathcal{C} = \{C_1,\ldots, C_q\}$ be a \cliqueWitness\ for $G-S$. The goal of this subsection is to split the cliques in $\CCC$ to levels such that 1) each level contains only bounded number of cliques (that are not singletons) and 2) if we do not remove any edge at level $i$, then we do not need to remove any edge at level $j>i$. We will later show that we do not need to remove any edges in cliques in level $5$. The following lemma is useful to define/bound the number of cliques at the first level, i.e., cliques that interact with $S$.

\begin{lemma}\label{lem:neighborhoodOfS}
	For every vertex $v\in S$ there are at most two cliques $C_1,C_2 \in \CCC$ such that $v$ is adjacent to all vertices in $C_1\cup C_2$ and to at most $6$ vertices in $V(G)\setminus (S\cup C_1\cup C_2)$. 
\end{lemma}
\begin{proof}
	By the choice of the set $S$, it follows that $G-(S\setminus \{v\})$ is a line graph. Let $\CCC'$ be \cliqueWitness\ for $G-(S\setminus \{v\})$. By definition, there are at most two cliques $C_1'$ and $C_2'$ in $\CCC'$ that contains $v$ and all its neighbors. If $|C_i'|\ge 5$, for some $i\in \{1,2\}$, then by Observation~\ref{obs:MaximalCliques2}, $C_i'\setminus \{v\}$ is a clique in $\CCC$ and we can set $C_i$ to be $C_i'\setminus \{v\}$. Else $|C_i'\setminus v|\le 3$ and $C_i'$ contributes to at most $3$ neighbors of $v$ in $G-S$. 
\end{proof}

The following lemma shows that cliques of size at least $k+7$ can serve as kind of
separators that will never be changed by a solution of size at most $k$. Hence, we can
remove all cliques separated from $S$ by large cliques. Moreover, it allows us to define the $(i+1)$-st level by only considering the cliques of size at most $k+6$ at level $i$.

\begin{lemma}\label{lem:replaceLarge}
	Let $C\in \CCC$ such that $|C|\ge k+7$ and let $A\subset E(G)$ be an optimal solution for $G$. 
	Then $A\cap E(G[C])=\emptyset$. 
Moreover, the \cliqueWitness\ $\CCC'$ for $G - A$ contains a clique $C'$ such that $C'\setminus S=C$.
\end{lemma}

\begin{proof}
	Let $\{u,v\}\in A$ such that $\{u,v\}\subseteq C$. Clearly there are at most $k-1$ vertices $w$ in $C$ such that either $\{u,w\}\in A$ or $\{w,v\}\in A$. Let $x\in C$ be such that $xv, xu$ are edges in $G- A$. Similarly, there are at most $k-1$ non-edges to $u,v,x$ in $G-A$, so let $y\in C$ be a vertex such that $yu, yv, yx$ are edges in $G- A$. Repeating the same argument once again, there is $z\in C$ such that $zu, zv, zx, zy$ are edges in $G- A$. However, the subgraph of $G-A$ induced on $u,v,x,y,z$ is $K_5$ minus an edge, which is one of the forbidden induced subgraphs in the characterization of line graphs. 
	
	The moreover part follows from the following argument. Since $|C|\ge k+7\ge 4$ and, by Observation~\ref{obs:MaximalCliques2} it follows that the \cliqueWitness\ $\CCC'$ contains a maximal clique $C'\supseteq C$. 
	It remains to show that no vertex in $V(G)\setminus (S\cup C)$ is in $C'$. Every vertex in $V(G)\setminus S$ is in two cliques $C_1$, $C_2$ in $\CCC$ that cover all its incident edges in $G-S$. If none of these two cliques is $C$, then $C$ intersect each of these two cliques in at most $1$ vertex. It follows that, because $|C|\ge 3$, there is no vertex in $V(G)\setminus (S\cup C)$ adjacent to all vertices~of~$C$.     
\end{proof}

Let us now partition the cliques in $\CCC$ into two parts $\CCC_{< k+7}$ and $\CCC_{\ge k+7}$ such that $\CCC_{< k+7}$ contains precisely all the cliques in $\CCC$ with less than $k+7$ vertices and $\CCC_{\ge k+7}$ contains the remaining cliques. We will refer to the cliques in $\CCC_{< k+7}$ as \emph{small} cliques and the cliques in $\CCC_{\ge k+7}$ as \emph{large} cliques. Intuitively, if we are forced to delete some edge in $G-S$, then this change had to be propagated from $S$ only by changes in small cliques.

We are now ready to define the level structure on the cliques in $\CCC$. We divide the cliques in $\CCC$ into levels $\LLL_1,\LLL_2, \ldots, \LLL_p$, for some $p\in \mathbb{N}$, that intuitively reflects on how far from $S$ the clique $C\in \CCC$ is if we consider a shortest path using only small cliques.
We will define the levels recursively as follows. By Lemma~\ref{lem:neighborhoodOfS} for every vertex $v \in S$ there exists at most two cliques $C_1,C_2\in \CCC$ such that $v$ is adjacent to all vertices in $C_1\cup C_2$ and to at most $6$ vertices in $V(G)\setminus (S\cup C_1\cup C_2)$. 
Now, for a vertex $v\in S$, let $\NNN^v$ denote the set of cliques that contains $C_1, C_2$ and all the cliques in $\CCC$ that contain at least one of the neighbors of $v$ in $V(G)\setminus (S\cup C_1\cup C_2)$.
We let $\LLL_1$ be precisely the set $\bigcup_{v\in S}\NNN^v$. Note that vertices in $C_1\cup C_2$ can each appear in one other clique that is not in $\NNN^v$ and in particular there are cliques that contain a vertex adjacent to a vertex in $S$ and are not in $\LLL_1$.
For $i>1$, we then let $\LLL_i$ be the set of cliques $C$ in $\CCC\setminus (\bigcup_{j\in \{1..i-1\}}\LLL_j)$ such that there is a small clique $C'$ in the previous level (\emph{i.e.}, $C'\in \LLL_{i-1}\cap \CCC_{< k+7}$) such that $C\cap C'$ is not empty. 

\begin{observation}\label{obs:withNeighborsInS}
	Let $C\in \CCC$ and $w$ a vertex in $C$. If $w$ has a neighbor in $S$, then either $C\in \LLL_1\cup \LLL_2$ or $w$ is in a large clique.
\end{observation}

\begin{proof}
	Let $v\in S$ be a neighbor of $w$. Then $\NNN^v\subseteq \LLL_1$ contains a clique $C'$ with $w\in C'$. Clearly $C'$ intersects $C$ in $w$. Hence either $C'$ is a large clique or by the definition of $\LLL_2$ the clique $C$ is in $\LLL_1\cup \LLL_2$. 
\end{proof}

Let $p\in \mathbb{N}$ be such that $\LLL_p\neq \emptyset$ and $\LLL_{p+1}=\emptyset$. While the following Reduction Rule is not completely necessary and would be subsumed by Reduction Rule~\ref{rrule:removeFarVertices}, we include it to showcase some of the ideas needed for the proof in a simplified setting. 

\begin{rrule}\label{rrule:largeSeparatorRemoval}
	Remove all vertices in $V(G)\setminus S$ that are not in a clique in $\bigcup_{i\in [p]}\LLL_i$. 
\end{rrule}

\begin{proof}[Proof of safeness]
	Let $H$ be the resulting graph and let $\CCC_H$ be a set of cliques of $H$ obtained from $\CCC$, by taking all cliques in $\bigcup_{i\in [p]}\LLL_i$ and for every clique in $C\in (\CCC\setminus \bigcup_{i\in [q]}\LLL_i)$, $\CCC_H$ contains $C\cap V(H)$, if it is nonempty. Since $H$ is an induced subgraph of $G$ and line graphs can be characterized by a set for forbidden induced subgraphs, it follows that for every $A\in E(G)$, if $G-A$ is a line graph, then $H-A$ is a line graph. It remains to show that if there is a set of edges $A\in E(H)$ such that $|A|\le k$ and $H-A$ is a line graph, then $G-A$ is also a line graph. Let $A$ be such a set of edges of minimum size and let $\CCC_A$ be a \cliqueWitness\ for $H-A$. It suffices to show that for every clique in $C\in (\CCC_H\setminus \bigcup_{i\in [p]}\LLL_i)$, it holds that $C\in \CCC_A$. If this is the case, we get a \cliqueWitness\ for $G-A$ by replacing the cliques of $\CCC_H\setminus \bigcup_{i\in [p]}\LLL_i$ in $\CCC_A$ by $\CCC\setminus \bigcup_{i\in [p]}\LLL_i$.
	
	Now, $C\in (\CCC_H\setminus \bigcup_{i\in [p]}\LLL_i)$ means that all cliques intersecting $C$ are large. Moreover, because all vertices in $H$ are in some clique on some level, by Lemma~\ref{lem:replaceLarge}, for each clique $C_1 \in \CCC_H$ that intersect $C$ there is a clique in $C'_1 \in \CCC_A$ that is the union of $C_1$ and some vertices in $S$. Hence, all vertices in $C$ are already in at least one clique in $\CCC_A\setminus C$ and all the edges incident to exactly one vertex in $C$ are already covered by these cliques. And hence every clique that contains a vertex in $C$ and intersects every other clique in $\CCC_A$ in at most one vertex has to be a subset of $C$. Moreover, the cliques in $\CCC_A$ that are subsets of $C$ have to be vertex disjoint, since every vertex is in at most $2$ cliques in $\CCC_A$. Hence,  if $C$ is not in $\CCC_A$, then some of the edges in $C$ have to be in $A$, but replacing all the subsets of $C$ in $\CCC_A$ by $C$ gives a \cliqueWitness\ for $H-A'$ for some $A'\subsetneq A$ which contradicts the fact that $A$ is of minimum size.
\end{proof}

We will also say that $C\in \CCC$ is at \distance\ $d$ from $S$, denoted by $\distC(C)$, if $C$ is in $\LLL_d$. We note that $\CCC$ still contains some cliques that are not in any of $\LLL_i$'s. We will let $\distC(C)=\infty$ for such a clique $C$. We can now upper bound the number of cliques at \distance\ $d$ from $S$.
 
 \begin{lemma}\label{lem:cliquesAtDistanceD}
 	There are at most $14|S|(k+6)^{d-1}$ cliques in $\CCC$ at level $d$, \emph{i.e.}, in $\LLL_d$.
 \end{lemma}
\begin{proof}
	By the definition of $\LLL_1=\bigcup_{v\in S}\NNN^v$, where $\NNN^v$ denote the set of cliques that
	 contains $C_1, C_2$ and all the cliques in $\CCC$ that contain at least one of the neighbors of $v$ in  
	 $V(G)\setminus (S\cup C_1\cup C_2)$. 
	 By Lemma~\ref{lem:neighborhoodOfS} for every vertex $v \in S$ there exists at most two cliques $C_1,C_2\in \CCC$ such that $v$ is adjacent to all vertices in $C_1\cup C_2$ and to at most $6$ vertices in $V(G)\setminus (S\cup C_1\cup C_2)$. Since every vertex appears in two cliques of $\CCC$, it follows that $|\NNN^v|\le 14$ and consecutively $\LLL_1$ contains at most $14|S|$ cliques.
	Now by the definition of $\LLL_d$ we know that for any $d\ge 2$ a clique is at level $d$ if and only if it shares a vertex with a small clique at level $d-1$. Since no three cliques in $\CCC$ can share a vertex the number of cliques at level $d$ is at most the number of vertices in the small cliques at level $d-1$ and the lemma follows by a simple induction~on~$d$.
\end{proof}

  The remainder of the algorithm consists of two steps. First, in Section~\ref{sec:boundingDistance}, we show that we can remove 
  all edges from cliques that are at \distance\ at least $5$ from $S$. Afterwards, 
  due to Lemma~\ref{lem:cliquesAtDistanceD}, we are left with only $\bigO{k^4}$ non-singleton cliques in $\CCC$. To finish the algorithm in Section~\ref{sec:boundingSize}, for each clique $C\in \CCC$
  that is not a singleton, we mark an arbitrary subset of $k+7$ vertices in $C$ and remove all unmarked vertices from $G$. It is then rather straightforward consequence of Lemma~\ref{lem:replaceLarge} 
  that this rule is safe and we get an equivalent instance with $\bigO{k^5}$ vertices.  

\section{Bounding the Distance from $S$}\label{sec:boundingDistance}

The purpose of this section is to show that it is only necessary to keep the cliques in $\CCC$ that are at \distance\ at most $4$ from $S$ (and adding a singleton for vertices covered by exactly one clique at \distance\ at most $4$). To do so, we need to show that there is always a solution that does not change the cliques at \distance\ $5$ at all. For this purpose, we first need to understand the interaction of cliques at \distance\ $4$ from $S$ with the solution. The first step will be to show that there is an optimal solution $A$ with \cliqueWitness\ $\CCC_A$ such that all cliques in $\CCC_A$ that share an edge with a clique in $\CCC$ at \distance\ at least $4$ from $S$ are actually subcliques of a clique in $\CCC$ (when restricted to $G-S$). It is a simple consequence of Lemma~\ref{lem:replaceLarge} that this is true for any clique that intersect a large clique in an edge. Hence, we can only care about cliques in $\CCC_A$ that intersect a small clique $C$ in an edge. By Observation~\ref{obs:withNeighborsInS}, no vertex in $C$ has a neighbor in $S$. It then follows by Observation~\ref{obs:MaximalCliques3} that any clique in $\CCC_A$ that intersects $C$ in an edge and is not a subclique of a clique in $\CCC$ is indeed a triangle. This leads us to the following definition.

\begin{definition}[bad triangle]
Let $A \subseteq E(G)$ be such that $G-A$ is a line graph and let $\mathcal{C}_{A}$ be a \cliqueWitness\ of $G-A$. A triangle $xyz \in \mathcal{C}_{A}$ is said to be \emph{bad} if it is not a sub-clique of a clique in $\CCC$, and one of the edges of the triangle, say $xy$, is an edge contained in a clique of $\mathcal{L}$-distance at least $4$ from $S$. 
\end{definition}

\begin{lemma}\label{lem:noBadTriangles}
	There exists an optimal solution without any bad triangle. 
\end{lemma}

\begin{proof}
	Let $A$ be an optimal solution and $\mathcal{C}_{A}$ the clique partition witness of $G-A$.
	Suppose $xyz$ is a bad triangle and let $C_1, C_2$ and $C_3$ be the elements of $\CCC$ containing the edges $xy, yz$ and $zx$ respectively. See also Figure~\ref{fig:noBadTriangle} for an illustration.
	Since $xyz$ is a bad triangle, no clique in $\CCC_A$ is a superset of $C_i$, $i\in \{1,2,3\}$ and it is a simple consequence of Lemma~\ref{lem:replaceLarge} that $C_i$ is a small clique. By definition of bad triangle, at least one of $C_1$, $C_2$, and $C_3$ is at \distance\ at least $4$ from $S$ and hence all of these cliques are at \distance\ at least $3$ from $S$. Let $X$ (resp. $Y$, $Z$) denote the other clique of $\mathcal{C}_{A}$ containing $x$ (reps. $y$, $z$). Let us define $X_1 = X \cap C_1, X_3 = X \cap C_3, Y_1 = Y \cap C_1, Y_2 = Y \cap C_2, Z_3 = Z \cap C_3$ and $Z_2 = Z \cap C_2$.
	
	Let $C'_1 = X_1 \cup Y_1$, $C'_2 = Y_2 \cup Z_2$ and $C'_3 = Z_3 \cup X_3$. Note that $C'_i$ is a sub-clique of $C_i$ for $i \in [3]$. Now for every $i \in [3]$ we will update $C'_i$ as follows. As long as there exists an edge $e$ in $C'_i$ such that $e$ belongs to $K_i \in \mathcal{C}_{A}$,  $K_i$ is a sub-clique of $C_i$ and $K_i  \not \subseteq C'_i$, we set $C'_i := C'_i \cup K_i$ (see also Figure~\ref{subfig:replacementClique}). When this process stops, $C'_i$ corresponds to the union of a set of elements of $\mathcal{C}_{A}$ : $K^i_1,\dots,  K^i_{l_i}$ which are sub-cliques of $C_i$, and $C'_i$. Moreover, for any edge $e$ of $C'_i$ which is strictly contained in another clique of $\mathcal{C}_{A}$ (meaning this clique is not $e$), then this clique has to be a triangle by Observation \ref{obs:MaximalCliques3}, as the clique of $\CCC$ containing $e$ is $C_i$. Let $e^i_1, \dots, e^i_{s_i}$ denote the set of such edges and let $C^i_1, \dots, C^i_{s_i}$ be the triangles of $\mathcal{C}_{A}$ containing these edges. Note that $|A \cap C'_1| \geq s_1$, as for any edge $e^1_j$, either $x$ or $y$ has to be non adjacent to each extremity in $G-A$ or the edge would be in two cliques of $\mathcal{C}_{A}$  (the same statement is also correct for $|A \cap C'_2|$ and $|A \cap C'_3|$ ). Let $A'$ be the set obtained from $A$ by 
	\begin{itemize}
		\item Removing all the edges of $A \cap C'_1$, $A \cap C'_2$ and $A \cap C'_3$.
		\item Adding one of the two edges of $C^i_j$ different from $e^i_j$ for every $i \in [3]$ and $j \in [s_i]$ (see Figure~\ref{subfig:pendantTriangles} illustrating the replacement of $C_i^j$ in $\CCC_A$ by its proper subclique in $\CCC_{A'}$ implied by this addition of an edge in $A'$.).
	\end{itemize}
	
	\begin{claim}
		$A'$ is a set of edges not larger than $A$ and such that $G-A'$ is a line graph with fewer bad triangles than $G-A$. 
	\end{claim}
	
	\begin{proof}
		The fact that $|A'|\leq |A|$ follows from the fact that  $|A \cap C'_i| \geq s_i$ for all $i \in [3]$. To see that $G-A'$ is a line graph, let us show that $\mathcal{C}_{A'}$ defined as follows is a clique partition witness for $G-A'$. Let $\mathcal{C}_{A'}$ be the set defined from $\mathcal{C}_A$ by
		\begin{itemize}
			\item Removing $C_A$, $X$, $Y$, $Z$, every $C^i_j$ for $i \in [3]$, $j \in [s_i]$, every $K^i_j$ for every $i \in [3]$ and $j \in [l_i]$ and every edge which are contained in one of the $C'_i$. 
			\item Adding $C'_i$ for $i \in [3]$ and for every $i \in [3]$ and $j \in [s_i]$ the edge of $C^i_j$ which has not been removed from $A$, as well as singletons for vertices belonging to only one clique.    
		\end{itemize}
		
		First it is clear that any set added to $\mathcal{C}_A'$ is a clique as $A'$ does not contain any edge in $A \cap C'_1$, $A \cap C'_2$ and $A \cap C'_3$ and these sets are cliques of $G$.
		
		Now take $B$ and $C$ two cliques of $\mathcal{C}_A'$. If $B$ and $C$ belong to $\mathcal{C}_A$, then clearly their intersection has size at most $1$. If one belongs to $\mathcal{C}_A$ and the other is the remaining edge of $C^i_j$ for $i \in [3]$ and $j \in [s_i]$, then it is also clear as it is true for $C^i_j$. For $i,j \in [3]^2$,  $C'_i$ and $C'_j$ also intersect on one vertex, because $C_i$ and $C_j$ do and moreover, the cliques of $\mathcal{C}_A$ intersecting $C'_i$ on two vertices are exactly the $C^i_j$, so if $B = C'_i$ and $A \in \mathcal{C}_A$, the intersection has also size at most $1$, and we covered all the cases for $|C \cap B|$. 
		
		Now for every vertex $x \in V(G)$, if $x$ does not belong to $C'_1, C'_2$ and $C'_3$, then it belongs to the same cliques as in $\mathcal{C}_A$ (where the $C^i_j$ have been reduced to an edge and a singleton). For the vertices of $C'_1, C'_2$ and $C'_3$ different from $x,y,z$, we replaced one sub-clique of $C_i$ by another. Finally $x$ belongs to $C'_1$ and $C'_3$, $y$ to $C'_1$ and $C'_2$ and $z$ to $C'_2$ and $C'_3$. 
		
		Suppose $uv$ is an edge of $E(G-A')$. If $uv$ belongs to one of the $C'_i$, then by definition of the $C^i_j$ and because we removed all these triangles, $uv$ only belongs to one clique. For the other edges of $E(G-A')$, the fact that $uv$ belongs to exactly one clique of $\mathcal{C}_A'$ follows from the fact that $A'$ differs on those edges from $A$ only because we added some edges of the $C^i_j$, and $\mathcal{C}_A$ differs on these vertices only because we changed $C^i_j$ into the remaining edge outside~$C'_i$.
		
		Overall $\mathcal{C}_{A'}$ is indeed a clique partition for $G-A'$. Moreover, to obtain it, we removed at least one bad triangle from $\mathcal{C}_A$ ($C_A$) without adding one. This ends the proof of the~claim.
	\end{proof}
	Finally, we can repeat the process until $\mathcal{C}_{A'}$ is without any bad triangles, which ends the proof of the lemma. 
\end{proof}

\begin{figure}[t]
	\begin{subfigure}[t]{0.30\textwidth}
		\centering
		\includegraphics[width=.85\textwidth,page=6]{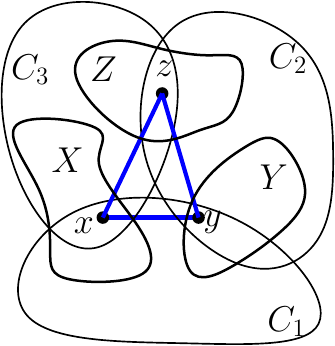}
		\caption{A bad triangle $xyz$ in $\CCC_A$. $C_1$, $C_2$, $C_3$ are three cliques in $\CCC$ containing $xy$, $yz$, and $xz$ respectively. $X,Y,Z\in \CCC_A$ are the cliques containing $x,y,z$ other than $xyz$. }\label{subfig:originalCliques}
	\end{subfigure}\hfill
	\begin{subfigure}[t]{0.30\textwidth}
		\centering
		\includegraphics[width=.85\textwidth,page=7]{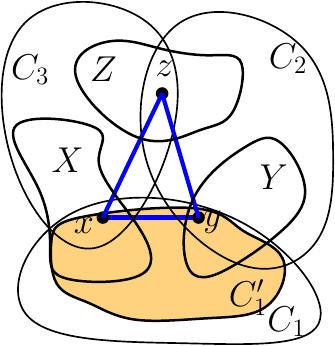}
		\caption{$C_1'$ is the inclusion minimal clique such that $(X\cup Y)\cap C_1\subseteq C_1'\subseteq C_1$ and for all $K\in \CCC_A$ if $K\subseteq C_1$ and $|K\cap C_1'|\ge 2$, then $K\subseteq C_1'$. $C_2'$ and $C_3'$ are defined analogously.}\label{subfig:replacementClique}
	\end{subfigure}\hfill
	\begin{subfigure}[t]{0.34\textwidth}
		\centering
		\includegraphics[width=.7\textwidth,page=8]{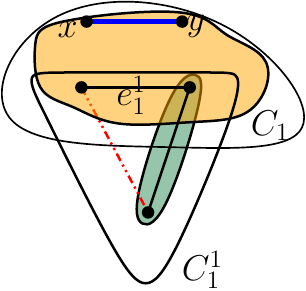}
		\caption{$C_1^1$ intersects $C_1'$ in an edge $e^1_1$. $C_1^1$ is replaced by an edge other than $e_1^1$. This forces to include one edge in $C_1^1$ to a solution $A'$. However, this can be seen as replacing an edge between $\{x,y\}$ and endpoints of $e_1^1$ that is in $A$.}\label{subfig:pendantTriangles}
	\end{subfigure}
	\caption{The treatment of bad triangles. Let $A\subseteq E(G)$ be an optimal solution, $\CCC_A$ a \cliqueWitness\ for $A$. A bad triangle $xyz$ together with cliques $X$, $Y$, $Z$, as defined in Subfigure~\ref{subfig:originalCliques} are replaced by cliques $C_1'$, $C_2'$, and $C_3'$ defined in Subfigure~\ref{subfig:replacementClique}. Subfigure~\ref{subfig:pendantTriangles} shows the treatment of cliques in $\CCC_A$ that intersect $C_i'$ in an edge. By definition of $C_i'$, such clique is not a subclique of $C_i$ and hence a triangle.}\label{fig:noBadTriangle}
\end{figure}

Before we show that indeed all cliques at \distance\ at least $5$ from $S$ are intact in some optimal solution, we show another auxiliary lemma that is rather simple consequence of Lemma~\ref{lem:noBadTriangles}, namely that there is a \cliqueWitness\ for some optimal solution $A$ such that no two cliques $\CCC_A$ that intersect the same clique $C\in \CCC$ at \distance\ at least $4$ from $S$ in an edge can intersect. This is important later to show that indeed no vertex in a clique $C\in \CCC$ at \distance\ $5$ from $S$ will be in two cliques in $\CCC_A$ that are not subsets of $C$.

\begin{lemma}\label{lem:oneVertex}
	There exists an optimal solution $A\subseteq E(G)$ without any bad triangles and \cliqueWitness\ $\CCC_A$ for $G-A$ such that for every $C\in \CCC$ of \distance\ at least $4$ and every $w\in C$, if $C^w_1$ and $C^w_2$ are the two cliques in $\CCC_A$ containing $w$, then either $C^w_1\cap C = \{w\}$ or $C^w_2\cap C = \{w\}$. 
\end{lemma}

\begin{proof}
	Let $A\subseteq E(G)$ be an optimal solution for $G$ without any bad triangles and \cliqueWitness\ $\CCC_A$ for $G-A$ minimizing the number of pairs $(C,w)$ for which $C$ is at \distance\ at least $4$, $w\in C$ and the two cliques, denoted $C^w_1$ and $C^w_2$, in $\CCC_A$ containing $w$ intersect $C$ in two vertices. Furthermore, it follows from Lemma~\ref{lem:replaceLarge} that $C$ is a small clique, as the clique containing $C$ as a subclique in $\CCC_A$ would intersect $C^w_1$ in two vertices. Since there are no bad triangles and $C$ is at \distance\ at least $4$, it follows that $C^w_1\subseteq C$ and $C^w_2\subseteq C$ and in particular $C^w_1\cup C^w_2$ is a clique in $G$. 
	Indeed, our goal is to replace $C^w_1$ and $C^w_2$ by a clique $D$ such that $C^w_1\cup C^w_2\subseteq D\subseteq C$. We start by setting $D=C^w_1\cup C^w_2$. We will also keep a track of cliques we will remove from $\CCC_A$. This set will be $\DDD$ and initialize it as $\DDD=\{C_1,C_2\}$.
	
	As in the proof of Lemma~\ref{lem:noBadTriangles}, the only reason why we cannot replace $C_1$ and $C_2$ by $D$ and obtain a solution that removes a subset of edges of $A$ is because there exist two vertices $v_1, v_2\in D$ and a clique $C_{12}\in \CCC_A$ with $\{v_1,v_2\}\subseteq C_{12}$. Observe that by our assumption there is no bad triangle and $C_{12}\subseteq C$. We let $D=D\cup C_{12}$ and $\DDD = \DDD\cup C_{12}$ and repeat until there is no such pair of vertices. Note that every vertex in $G$ is in at most two cliques of $\CCC_A$. Therefore, this process has to stop after at most $2|C|$ steps. 
	
	When there are no two vertices in $D$ that appear together in a different clique, we remove $\DDD$ from $\CCC_A$ and replace it by $D$ and $\{v\}$. For every vertex that appears in $D$, we removed one clique that it appeared in. Hence, every vertex appears in at most $2$ cliques and we can always add a singleton to \cliqueWitness\ for vertices that are only in one clique. Moreover, no two cliques intersect in two vertices, since $D$ is the only clique we added, and we removed/changed all the cliques that intersected $D$ in at least two vertices. Finally, all edges in $G-A$ remain covered, we only potentially covered some additional edges in $D$. 
	
	Note that this procedure does not introduce any bad triangles or new pair $(C',w')$ for which $C'$ is at \distance\ at least $4$, $w'\in C'$ and the two cliques in $\CCC_A$ containing $w'$ intersect $C'$ in two vertices. As it also removes one such pair, we obtain a contradiction with the choice of $A$. We can therefore deduce that $A$ does not contain such pair $(C,w)$ and the lemma follows.
\end{proof}

Finally, we can state the main lemma of this section. 

\begin{lemma}\label{lem:distance5cliques}
	There exists an optimal solution $A$ for $G$ and a \cliqueWitness\ $\CCC_A$ for $G-A$ such that for every clique $C\in \CCC$ at \distance\ at least $5$ it holds that $C\in \CCC_A$. 
\end{lemma}
\begin{proof}
	Let $A$ be an optimal solution without any bad triangles and \cliqueWitness\ $\CCC_A$ for $G-A$ such that for every $C\in \CCC$ of \distance\ at least $4$ and every $w\in C$, if $C^w_1$ and $C^w_2$ are the two cliques in $\CCC_A$ containing $w$, then either $C^w_1\cap C = \{w\}$ or $C^w_1\cap C = \{w\}$. Note that existence of such a solution is guaranteed by Lemma~\ref{lem:oneVertex}. Moreover let $(A,\CCC_A)$ be such an optimal solution satisfying properties in Lemma~\ref{lem:oneVertex} that minimizes the number of cliques $C\in\CCC$ of \distance\ at least $5$ such that $C\notin\CCC_A$. We claim that $A$ satisfies the properties of the lemma. 
	
	For a contradiction let $C\in \CCC$ be a clique at \distance\ at least $5$ and let $C_1,\ldots, C_p$ be the cliques in $\CCC_A$ that intersects $C$ in at least $2$ vertices.
	Since there is no bad triangle, it follows that $C_i\subseteq C$ for all $i\in [p]$ and by optimality of $A$, $p=1$ (else $\bigcup_{i\in[p]}C_i$ is missing at least one edge). We claim that $C=C_1$. Else let $v\in C\setminus C_1$. Note that $C$ is a small clique and hence by Observation~\ref{obs:withNeighborsInS} $v$ does not have a neighbor in $S$. In particular all neighbors of $v$ are covered by two cliques in $\CCC$, one of those cliques is $C$ and let the other clique be $C^v$. Moreover, Let $C^v_1$ and $C^v_2$ be the two cliques in $\CCC_A$ containing $v$. Since $v\in C\setminus C_1$ both $C^v_1$ and $C^v_2$ are subsets of $C^v$. However, $C^v$ is either a large clique and $\CCC_A$ contains $C^v$ and the cliques $C^v_1$ and $C^v_2$ are $C^v$ and $\{v\}$ respectively, or $C^v$ is a small clique, in which case $C^v$ is at \distance\ at least $4$ from $S$, because it shares a vertex with the clique $C$ at \distance\ at least $5$ from $S$. It follows by the choice of $A$ that either $C^v\cap C^v_1=\{v\}$ or $C^v\cap C^v_2=\{v\}$, but then again either $C^v_1$ or $C^v_2$ is the singleton $\{v\}$. However then the \cliqueWitness\ $(\CCC_A\setminus \{C_1,\{v\}\})\cup \{C_1\cup\{v\}\}$ defines a better solution. It follows that indeed $C\in \CCC_A$ for all cliques in $\CCC$ at \distance\ at least $5$ in $G$.   
\end{proof}

We are now ready to present our main reduction rule. Note that it would seem that we could remove just the vertices that do no appear in a clique at distance at most $4$. However, because of the large cliques in at the first four levels, we would be potentially left with many cliques at \distance\ infinity that we cannot remove because all of their vertices are in a large clique at \distance\ at most $4$ from $S$. While this case could have been dealt with separately, we can actually show a stronger claim, \emph{i.e.,} that we can remove all edges from $G$ that are covered by a clique at \distance\ at least $5$ from $S$. Note that in this case we cannot easily claim that if $(G,k)$ is YES-instance then so is the reduced instance and we crucially need the fact that cliques at \distance\ at least $5$ are kept in \cliqueWitness\ of some optimal solution.

\begin{rrule}\label{rrule:removeFarVertices}
	Remove all edges $uv\in E(G)$ such that $\{u,v\}\subseteq C$ for some clique $C$ with $\distC(C)\ge 5$. Afterwards remove all isolated vertices from $G$.	
\end{rrule}

Let $\DDD$ be the set of cliques at \distance\ at least $5$ from $S$, $V_5$ the set of vertices that appear in a clique in $\DDD$ and in a clique in $\CCC\setminus \DDD$ and $G'$ be the graph obtained after applying the reduction rule and let $\CCC'=(\CCC\setminus \DDD)\cup \bigcup_{v\in V_5}\{v\}$. Note that $\CCC'$ is a \cliqueWitness\ for $G'-S$ and that $\{v\}$, for $v\in V_5$, is a clique at \distance\ at least $5$.  

\begin{proof}[Proof of safeness]
Let $\DDD$, $V_5$, $G'$, $\CCC'$ be as described above and let $A$ be an optimal solution for $G'$, that is $G'-A$ is a line graph, and let $\CCC_A$ be \cliqueWitness\ for $G'-A$. By Lemma~\ref{lem:distance5cliques}, we can assume that $\bigcup_{v\in V_5}\{v\}\subseteq \CCC_A$. We will show that $(\CCC_A\setminus \bigcup_{v\in V_5}\{v\})\cup \DDD$ is a \cliqueWitness\ for $G-A$.
	Clearly each edge in $G-A$ is either covered by $(\CCC_A\setminus \bigcup_{v\in V_5}\{v\})$ or by $\DDD$. It is also easy to see that every vertex is in precisely two cliques. Moreover, two cliques in $\DDD$ intersect in at most $1$ vertex, because $\DDD\subseteq \CCC$ and similarly two cliques in $\CCC_A$ intersect in at most one vertex. Finally, let $D\in \DDD$ and $C\in (\CCC_A\setminus\bigcup_{v\in V_5}\{v\})$. Clearly, $D\cap C\subseteq V_5$. Moreover, for $\{u,v\}\subseteq D$, the edge $uv$ is not in $G'$ and hence $\{u,v\}\not\subseteq C$. Hence, $|D\cap C|\le 1$. 
	
	On the other hand, 
	let $A$ be an optimal solution for $G$ and a \cliqueWitness\ $\CCC_A$ for $G-A$ 
	such that for every clique $C\in \CCC$ at \distance\ at least $5$ it holds that 
	$C\in \CCC_A$.  Note that the existence of $(A,\CCC_A)$ is guaranteed by  
	Lemma~\ref{lem:distance5cliques}. We claim that $G'-A$ is a line graph. By the choice of $(A,\CCC_A)$, it follows that $\DDD\subseteq \CCC_A$. Moreover, for every edge $e$ that is covered by a clique in $\DDD$ it holds that $e\notin E(G')$. It follows rather straightforwardly that $\CCC_A\setminus \DDD\cup \bigcup_{v\in V_5}\{v\}$ is indeed a \cliqueWitness\ for $G'-A$. 
\end{proof}

\section{Finishing the Proof}\label{sec:boundingSize}
Suppose now that $G$, $S$, and $\CCC$ correspond to the instance after applying Reduction Rules~\ref{rrule:largeSeparatorRemoval}~and~\ref{rrule:removeFarVertices}. Clearly all cliques in $\CCC$ are either at \distance\ at most $4$ from $S$ or there are singletons at distance $5$ or infinity, depending on whether the singleton intersects a small or a large clique, respectively. It follows from Lemma~\ref{lem:cliquesAtDistanceD} that there are at most $\bigO{k^4}$ cliques at distance at most $4$. We let $M$ be any minimal w.r.t. inclusion set of vertices such that for every clique $C$ in $\CCC$ at \distance\ at most $4$ it holds that $|M\cap C|\ge \min\{|C|,k+7\}$. Such a set $M$ can be easily obtained by including arbitrary $\min\{|C|,k+7\}$ vertices from every clique $C$ at distance at most $4$ and then removing the vertices $v$ such that $|(M\setminus \{v\})\cap C|\ge \min\{|C|,k+7\}$ for all $C\in \CCC$ at \distance\ at most $4$. From this construction it is easy to see that $|M|=\bigO{k^5}$.

\begin{rrule}\label{rrule:largeCliques}
	Remove all vertices in $V(G)\setminus (S\cup M)$ from $G$.
\end{rrule}
\begin{proof}[Proof of safeness]
	Let the \cliqueWitness\ $\CCC'$ for $G-(S\cup M)$ be $\{C\cap M\mid C\in \CCC, C\cap M\neq\emptyset\}$.
	Since line graphs are characterized by a finite set of forbidden induced subgraphs, it is easy to see that if $G-A$ is a line graph, for some $A\subseteq E(G)$, then $G[S\cup M]-A=(G-A)[S\cup M]$ is also a line graph. For the other direction, 
	let $A\subseteq E(G)$ be such that $G[S\cup M] - A$ is line graph. We will show that $G- A$ is a line graph. Let $\CCC_A$ be a \cliqueWitness\ for $G[S\cup M]-A$. 
	Now let $\CCC'_A$ be the set we obtain from $\CCC_A$ by adding to it all the singleton cliques in $\CCC$ that do not contain a marked vertex and for every clique $C\in \CCC_A$ for which there exists $C'\in \CCC$ with $C\setminus S\subseteq C'$, we replace $C$ by $C'\cup (C\cap S)$. 
	
	First let us verify that every vertex in $V(G)$ is in precisely two cliques in $\CCC'_A$. It is easy to see that this holds for $v\in S\cup M$, because $\CCC_A$ is a \cliqueWitness\ for $G[S\cup M]-A$ and we only added new cliques containing vertices in $V(G)\setminus (M\cup S)$ or extended existing cliques in $\CCC_A$ by vertices in $V(G)\setminus (M\cup S)$. 
	Now let $v\in V(G)\setminus M$ and let $C_1,C_2\in \CCC$ be two cliques that contain $v$. Because all cliques in $\CCC$ at \distance\ at least $5$ are singletons and we keep all vertices of the cliques at \distance\ at most $4$ of size less than $k+7$, it follows that $C_1$ and $C_2$ either both contain at least $k+7$ vertices or one of them, say $C_2$, is a singleton and the other, $C_1$, contains at least $k+7$ vertices. If $C_2$ is a singleton, then $C_2\in \CCC'_A$. Else for $C_i$, $i\in \{1,2\}$, with $|C_i|\ge k+7$ there is $C_i'\in \CCC'$ with $|C_i'|\ge k+7$ and $C_i'\subseteq C_i$. 	By Lemma~\ref{lem:replaceLarge}, $\CCC_A$ contains a clique $C_i^A$ such that $C_i^A\setminus S= C_i'\setminus C_i$. By the construction of $\CCC'_A$ it now follows that $\CCC_A'$ contains $C_i^A\cup C_i$. From Lemma~\ref{lem:neighborhoodOfS} it follows that if $u\in S$ is adjacent to at least $7$ vertices in a clique in $\CCC$, then it is adjacent to the whole clique. Hence $C_i^A\cup C_i$ indeed induces a complete subgraph of $G-A$. It follows that $v$ is indeed in precisely two cliques in $\CCC_A'$. Note that above also shows that the sets in $\CCC_{A}'$ induce cliques in $G-A$. Furthermore every edge in $G-A$ either has both endpoints in $S\cup M$ and are covered by a clique $C$ in $\CCC_A$ such that $\CCC_A'$ contains a superset of $C$, or they are in the same clique of size at least $k+7$ in $\CCC$ that is a subset of a clique in $\CCC_A'$ as well. 
	
	 It remains to show that $|C_1\cap C_2|\le 1$ for all cliques in $\CCC'_A$.
	 If $|C_1\cap C_2|\ge 2$, then at least one of the vertices in $C_1\cap C_2$ 
	 has to be outside $S\cup M$. But then from the above discussion follows that 
	 $C_1\setminus S$ and $C_2\setminus S$ are in $\CCC$, $|C_1\setminus S|\ge 
	 k+7$, $|C_2\setminus S|\ge k+7$ and at least $k+7$ vertices from each of 
	 $C_1\setminus S$ and $C_2\setminus S$ are in $G[S\cup M]$. Clearly, 
	 $C_1\setminus S$ and $C_2\setminus S$ intersect in at most one vertex, let us 
	 denote it $u$, and the other vertices in the intersection of $C_1$ and $C_2$ 
	 are in $S$. Let $v$ be arbitrary vertex in $C_1\cap C_2\cap S$. Note that $v$ 
	 is adjacent to at least $7$ vertices in both $C_1\setminus S$ and 
	 $C_2\setminus S$ and by Lemma~\ref{lem:neighborhoodOfS} it is adjacent to all 
	 vertices in $(C_1\cup C_2)\setminus S$. Since $G-(S\setminus \{v\})$ is a line 
	 graph, it follows that $G[(C_1\cup C_2)\setminus (S\setminus \{v\})]$ is a 
	 line graph. Every vertex in $C_1\setminus (S\cup \{u\})$ is in exactly one 
	 other clique in $\CCC$. This clique intersects $C_2\setminus (S\cup \{u\})$ in 
	 at most one vertex. Therefore, there is a pair of vertices $w_1\in 
	 C_1\setminus (S\cup \{u\})$, $w_2\in C_2\setminus (S\cup \{u\})$ such that 
	 $w_1w_2\notin E(G)$. Now $uvw_1$ and $uvw_2$ are two odd triangles (any vertex 
	 in $C_i\setminus (S\cup \{u,w_i\})$ is adjacent to three vertices of the 
	 triangle $uvw_i$) that share a common edge, however $uvw_1w_2$ is not a $K_4$. 
	 Hence, $G[(C_1\cup C_2)\setminus (S\setminus \{v\})]$ is not a line graph, a 
	 contradiction. It follows that if two cliques in $\CCC$ of size at least $k+7$ 
	 intersect in a vertex in $G-S$, then no vertex in $S$ is adjacent to both 
	 cliques and consequently no two cliques in $\CCC'_A$ intersect in at least two 
	 vertices.   
	 
	 It follows that $\CCC_A'$ is indeed a \cliqueWitness\ for $G-A$ and by point (2) in Theorem~\ref{thm:LGcharacterization}, $G-A$ is indeed a line graph.
\end{proof}

We are now ready to prove Theorem~\ref{thm:main}.

\polykernelthm*

\begin{proof}
	We start the algorithm by finding the set $S$ of at most $6k$ vertices such that for every $v\in S$ the graph $G-(S\setminus \{v\})$ is a line graph. 
	This is simply done by greedily finding maximal set of pairwise edge-disjoint forbidden induced subgraphs. 
	Afterwards, we construct a \cliqueWitness\ $\CCC$ for $G-S$ by using the algorithm of Lemma~\ref{lem:constructingWitness}.
	Finally, we apply Reduction Rules~\ref{rrule:largeSeparatorRemoval}, \ref{rrule:removeFarVertices}, and \ref{rrule:largeCliques} in this order. By the discussion above Reduction Rule~\ref{rrule:largeCliques}, after applying all the reduction rules, the resulting instance has $\bigO{k^5}$ vertices. The correctness of the kernelization algorithm follows from the safeness proofs of the reduction rules.
\end{proof}

\section{Concluding Remarks}

In this paper, we positively answered the open question from WorKer 2013 about kernelization of \ldelete\ by giving a kernel for the problem with $\bigO{k^5}$ vertices. Our techniques crucially depend on the structural characterization of the line graphs. We believe that similar techniques could lead also to polynomial kernels for \textsc{Line-Graph-Edge Addition} and \textsc{Line-Graph-Edge Editing}. In particular, a result similar to Lemma~\ref{lem:replaceLarge} still holds when we allow addition of the edges. However, we were not able to bound the distance from $S$. Main difficulty seems to be the possibility of merging of some small cliques into one in a \cliqueWitness. It is also worth noting that the line graphs of multigraphs (\emph{i.e.,} graphs that allow multiple edges between the same pair of vertices) have a similar structural characterization with the main difference being that the cliques in a \cliqueWitness\ can intersect in more than just one vertex. The kernelization of the edge deletion (as well as addition or editing) to a line graph of a multigraph remains open as well. Finally, the kernelization of \textsc{Claw-free Edge Deletion} as well as of the edge deletion to some of the other natural subclasses of claw-free graphs remain wide open.

\bibliography{paw-free}
\end{document}